\documentclass[12pt]{article}
\usepackage{amsmath,amsfonts,amssymb,amsthm,amsbsy}
\usepackage{amsfonts}
\usepackage{amssymb}
\usepackage{epsfig}
\usepackage{color,url}
\usepackage[round]{natbib}
\usepackage{subfigure}
\DeclareMathSizes{12}{12}{10}{10} 
\definecolor{Red}{rgb}{1,0,0}

\newcommand{\BEAS}{\begin{eqnarray*}}
\newcommand{\EEAS}{\end{eqnarray*}}
\newcommand{\BEA}{\begin{eqnarray}}
\newcommand{\EEA}{\end{eqnarray}}
\newcommand{\BIT}{\begin{itemize}}
\newcommand{\EIT}{\end{itemize}}
\newcommand{\BNUM}{\begin{enumerate}}
\newcommand{\ENUM}{\end{enumerate}}

\newtheorem {Thm}{Theorem}
\newtheorem {prop} {Proposition}

\newcommand{\bp}{\begin{prop}}
\newcommand{\ep}{\end{prop}}

\def \beq {\begin{equation} }
\def \eeq {\end{equation} }

\def \bea {\begin{eqnarray} }
\def \eea { \end{eqnarray} }
\def \beas {\begin{eqnarray*} }
\def \eeas { \end{eqnarray*} }

\def \Var{{\rm Var}}

\def \E {{\rm E}}

\def \l({\left(}
\def \r){\right)}

\def \bi {\begin{itemize}}
\def \ei {\end{itemize}}

\begin{document}

\title{A semiparametric mixture method for local false discovery rate estimation}
 
\author{Seok-Oh Jeong\\ Hankuk University of Foreign Studies \footnote{seokohj@hufs.ac.kr} \and Dongseok Choi   \\Oregon Health and Science University\footnote{choid@ohsu.edu} \and Woncheol Jang \\ Seoul National University \footnote{wcjang@snu.ac.kr}}
\date{\today}
\maketitle
\vspace*{-0.3in}

 \begin{abstract}
We propose a semiparametric mixture model to estimate  local false discovery rates in multiple testing problems. The two pilars of the proposed approach are Efron's empirical null principle and log-concave density estimation for the alternative distribution. Compared to existing methods, our method can be easily extended to high dimension. Simulation results show that our method outperforms other existing methods and we illustrate its use via case studies in astronomy and microarray.

\vspace{.2in}

\medskip
\noindent {\it Keywords and Phrases: Brown-Stein model, false discovery rate, log-concave, mixture model, semiparametric model}  

\end{abstract}

\section{Introduction}
We consider a specialization of the Brown-Stein model used by \citet{Efron08a}. Suppose that $(\delta_1, Z_1), \ldots, (\delta_N, Z_N)$ are independently  generated by the following hierarchical model:
\begin{eqnarray}\label{mixture1}
\delta &\sim& g(\delta) = p_0 g_0 ( \delta) + (1-p_0) g_1 (\delta) \\
Z | \delta &\sim& f(z | \delta) = \phi_{\delta, 1} (z) 
\end{eqnarray}
where  $g_0 (\delta) = \phi_{\mu, \sigma^2} (\delta) $ and  $\phi_{\mu, \sigma^2}$ is the normal probability density function with mean $\mu$ and variance $\sigma^2$. Here we assume the probability $p_0$ and   the probability density function $g_1$  are unknown.

The  marginal distribution of $Z$ is given as follows:
$$
f(z) =    \int \phi (z -\delta ) g( \delta ) d\delta = p_0 \phi_{\mu, \tau^2} (z) + (1-p_0)  f_1 (z)
$$
 where $\phi ( \cdot) = \phi_{0,1}( \cdot)$, $\tau^2 = \sigma^2+1$ and  $f_1 (z)  =  \int \phi(z -\delta ) g_1 ( \delta  ) d\delta$.   

The model appears in three contexts: (i) in multiple testing problems (fMRI, microrarray), probit transformed $p$-values  under $H_0$ follows  the standard normal distribution while the mariganl probability density function of  the probit transformed  $p$-values associated with $H_1$ is unknown  \citep{Efron08a};   (ii) in variable/basis selection,  a mixture prior is used to achieve sparsity \citep{Johnstone05}; (iii) in prediction, Fisher's discriminant function can be regularized using a connection with the false discovery rate theory \citep{Efron08b}.  Due to close connections between the false discovery rate and all aforementioned problems, we will explain our model in the framework of  multiple testing through this paper. 

Suppose that  we observe $N$ cases each with some effect size $\delta_i$  from either null ($\delta_i =0$) or alternative ($\delta_i \ne 0$) with the prior probability $p_0= \Pr \{\mbox{null}\}$ or $p_1 =\Pr \{ \mbox{alternative}\} = 1-p_0 $ and with $z$-values having density $f$:
$$
f(z) =    p_0 f_0 + (1-p_0)  f_1 (z).
$$
Here $Z$ can be either test statistics or probit-transformed $p$-values.

In multiple testing, adjusting for multiplicity is of great interest, To do so,  we may consider to control for either the local false discovery rate, $fdr(z) = \Pr (\delta_i =0\mid z_i )$  or  the False Discovery Rate,  $FDR(z) = \Pr ( \delta_i =0  \mid Z_i > z )$ assuming one-sided testing. In this paper, we focus on estimating $fdr(z)$ that can be viewed as the posterior probability of a case being from the null given $z$.  Under the mixture model above, it is straightforward to show
$$
fdr(z) = p_0 f_0 (z)/f (z). 
$$

It is natural to assume that $f_0$ follows the standard normal distribution. However,  \citet{Efron08a} suggested that theoretical null  distribution $N(0,1)$ may not be suitable for $f_0$ and proposed to  estimate $f_0$ with the {\em empirical null distribution} $N( \mu, \tau^2)$ where $\mu$ and $\tau$ are to be estimated from data. Furthermore, he used  the {\em zero assumption} to estimate $\pi_0$ and employed {\it Lindsey's method} to estimate the marginal density $f$.  Note that his estimates  may not follow the original mixture structure since the estimates are given from separate procedures. In other words, the mixture only used to define $fdr$, but not for making inference of $fdr$. 
 
In this paper, we propose a semiparametric mixture model that cover a rich class of distributions for alternative and estimate $\pi_0, f_0, f_1$ simultaneously based on the proposed mixture model.  Furthermore, our method can be easily extended to  high dimension.

The remainder of the paper is organized as follows.  Section 2 introduces  log-concave densities and our semiprarametric mixture model. In Section 3, we present numerical studies to show that our method outperforms other  existing methods.  Section 4 presents  case studies including estimating distributions of radial velocities in astronomy  and  ophthalmologic  gene expression data analysis. Section 5 concludes this paper.

\section{Semiparametric Mixture Model}

\subsection{Mixture model with smoothed log-concave}
In this section, we propose  a semiparametric  mixture model for $f$.  We assume that $g_1$ belongs to log-concave densities.  The semiparametric mixture prior results in another semiparametric mixture distribution for marginal distribution of $z$.
 
 It is reasonable to assume that  the alternative distribution for $p$-values is log-concave and defined on $(0, \frac{1}{2})$. The following theorem justifies the choice of log-concave as an alternative distribution for profit transformed $p$-values.

 \begin{Thm}
 Let $H(t)$ and $h(t)$ denote the {\sf cdf} and {\sf pdf} of the alternative distribution of $p\mbox{-value}=Pr (Z > z)$.  Suppose that $h(t)$ is log-concave and continuous differentiable and  $(0,\frac{1}{2})$.  Define $Z=\Phi^{-1}(1-p)$, probit-transformed $p$-value.  Then the {\sf pdf} of $Z$ is log-concave under the alternative.
  \end{Thm}
 \begin{proof}
 By Theorem 1 in \citet{Bagnoli05}, $H(t)$ is  also log-concave on $(0,\frac{1}{2})$.  Let $F_1(z)$ denote {\sf cdf}  of the alternative distribution of $Z$. Then
 $$
 F_1(z) = {\sf Pr} (Z  \le z)  = {\sf Pr} ( \Phi^{-1} (1-p) \le z)  = {\sf Pr} ( 1-p \le \Phi (z) ) = 1-H( 1-\Phi (z))
 $$
  
 Therefore, the {\sf pdf} of $z$ is given as follows:
 $$
 f_1(z) = \frac{d}{dz} F_1(z) = h(1-\Phi(z)) \phi(z).
 $$
Note that $\Phi(z)$ is convex on ${\mathbb R}^{-}$.  Since $h$ is log-concave, $h(1-\Phi(z))$ is log-concave by Theorem 7 in \citet{Bagnoli05}.  Furthermore, a product of log-concave functions is also log concave. Therefore $f_1(z)$ is log-concave.
 \end{proof}

 Advantages of the class of log-concave densities that is a subset of unimodal densities are well studied in \citet{Walther02, Walther09}. \citet{Walther02} showed the existence of  the nonparametric MLE of a univariate log-concave density that can be computed via an efficient algorithm such as  an active set algorithm and an iterative convex minorant algorithm \citep{Dumbgen11}. However, the MLE  may have sharp discontinuities at the boundary of its support as well as sharp kinks. To remedy this issues, \citet{Dumbgen09} introduce a smoothed log-concave $\tilde f_1$, given by 
 $$
 \tilde f_1 (t) = \int  \phi_{z, a^2} (t) \hat f_1 (z) dz.
 $$
 for some bandwidth $a >0$. Here $\hat f_1$ is the nonparametric MLE of $f$. Since it is convolution of two log-concave densities, the smoothed log-concave also belongs to the class of log-concave
 densties. 
 
 Let $\tilde F$ and $\hat F$ be the distribution functions of $\tilde f$ and $\hat f$ respectively.  Since $\tilde f$ is convolution of $N(0, a^2)$ and $\hat f$, it is easy to show 
 $$
 \Var (\tilde F) = a^2 + \Var (\hat F) . 
 $$
 
 Based on the above equation,  \citet{Dumbgen09} suggest to choose the bandwidth 
 $$
 a = \sqrt{ \Var(\hat F) - \Var(\tilde F)},
 $$
 where
 $$
 \Var(\hat F)  =\frac{1}{N-1}\sum_{i=1}^N ( z_i - \bar{z} )^2, \quad \Var(\tilde F) = \int (z - \bar z)^2 \hat f(z) dz,
 $$
 with  $\bar{z} = \sum_{i=1}^N z_i /N$.
 
Hence,  the proposed semiparametric estimator of $f$ is 
$$
\hat p_0 \phi_{\hat \mu, \hat \tau^2} (z) + (1-\hat p_0) \tilde f_1 (z) .
$$

\subsection{EM Algorithm}

It is natural to consider an EM-type algorithm to fit the proposed semiparametric model. To do so, we first define the likelihood function:
\begin{eqnarray*}
\ell (\mu, \tau^2, z, \Delta) &=&  \sum_{i=1}^N \left[ \Delta_i \log \phi_{\mu, \tau^2} (z_i) + ( 1- \Delta_i ) \log f_1 (z_i) \right] \\&&+ \sum_{i=1}^N \left[\Delta_i \log p_0 + (1-\Delta_i) \log (1-p_0)\right]
\end{eqnarray*}
where $\Delta$ is a latent variable indicating the group membership. 

Following \citet{Chang07}, we first run the EM algorithm for a Gaussian mixture  to get initial values of $\mu$ and $\tau$. Then the EM algorithm for the proposed semiparametric model is given below: 
\begin{enumerate}
 \item E-step: compute the posterior probability: 
 $$\gamma_i = \E ( \Delta_i | \mu, \tau^2, z)= \frac{\hat p_0 \phi_{\hat \mu, \hat \tau^2} (z_i) } {\hat p_0 \phi_{\hat \mu, \hat \tau^2} (z_i) + (1-\hat p_0)\tilde f_1 (z_i)}$$ 

\item M-step: compute the smoothed log-concave estimator  $\tilde f_1$ based on $z_i$ with weights $\gamma_i$   and
$$
\hat \mu = \frac{ \sum_{i=1}^N \gamma_i z_i}{\sum_{i=1}^N  \gamma_i}, \quad \hat\tau^2 = \frac{ \sum_{i=1}^N \gamma_i( z_i - \hat \mu)^2}{\sum_{i=1}^N \gamma_i}, \quad \hat p_0  = \frac{1}{N} \sum_{i=1}^N \gamma_i .
$$
\end{enumerate}
To compute the smoothed log-concave MLE $\tilde f$ in M-step, {\sf R} package {\sf logcondens} \citep{Dumbgen11} is used.  


\subsection{Extension to Multivariate Cases} 

Suppose one observe multivariate statistics $\mathbf{Z}=(Z_1, \ldots, Z_k)$ and each statistic provides unique information for data. It is easy to extend the concept of the $fdr$ to high dimension:
$$
fdr(\mathbf{Z}) = p_0\frac{f_0 (\mathbf{Z})}{f(\mathbf{Z})} 
$$
However, applying Efron's method to multivariate is not straightforward. For example, Lindsey's method is no longer valid to estimate the null distribution even for 2-dimension. 

\citet{Ploner06} proposed to extend the local {\sf fdr}   to 2-dimension to address the problem of gene expression data with small variance. Specifically, they computed 2 dimensional  local {\sf fdr}  with  $t$-test statistics and standard error estimates for each gene. They  estimated $f_0$ with discrete smoothing of binomial data after binning the data.  However, their method still requires a smoothing parameter and has an issue with boundary bias. On the other hand, our method is straightforward to extend to multivariate and is not involved in choosing a smoothing parameter. 

\citet{Cule10} extended \citet{Dumbgen09} to multivariate settings and implemented multivariate smoothed log-concave estimation in  {\sf R} package {\sf logConcDEAD} \citep{Cule09}.

  \section{Simulation Studies}

In this section we investigate the performance of the proposed semiparametric approach with simulation studies where we can compare the results with Efron's method. We conduct $M=500$ Monte Carlo experiments with the sample size $N=1,000$ and  use  three measures of performance: empirical False Discovery Rate (FDR), empirical  False Nondiscovery Rate (FNR) and the Root Mean Squared Error (RMSE) which is given by 

$$
{\rm RMSE} = \frac{1}{M}\sum_{r=1}^M \sqrt{{\sum_{i: {\rm fdr}(z_i^{()r})\le 0.5} \left\{\widehat{\rm fdr}(z_i^{(r)})-{\rm fdr}(z_i^{(r)})\right\}^2}\left/{\sum_{i=1}^N I({\rm fdr}(z_i^{(r)})\le 0.5)}\right.}.
$$
where $z_1^{(r)}, z_2^{(r)},\cdots, z_N^{(r)}$ is the $r$-th random sample generated from the Monte Carlo simulation. Here, we only consider ${\rm fdr}(z_i^{(r)})\le 0.5$ in the definition of the RMSE because the accuracy of $fdr$ estimation is required mainly for the region where the $fdr$ is small.  For example, we are usually interested in the behavior of $fdr$ estimates where ${\rm fdr}(z)\le 0.2$.

\begin{table}[t]
  \begin{center}
  \caption{\label{model} Simulation scenarios:  $*$ indicates the convolution of two densities. ${\rm Gamma}_2(\alpha, \beta)$ denotes the bivariate gamma distribution of $(V_1, V_2)\stackrel{iid}\sim {\rm Gamma(\alpha, \beta)}$.}
   \begin{tabular}{lccccccc}
  \hline
    \multicolumn{2}{c}{Scenario} & & $p_0$ & & $f_0$ & & $f_1$ \\
    \hline
    Univariate & 1 & & 0.95  & & $N(0,1)*N(0,10^{-6})$ & & $N(0,1)*N(3.5, 0.5)$ \\
    & 2 & & 0.90  & & $N(0,1)*N(0,10^{-6})$ & & $N(0,1)*N(3.5, 0.5)$ \\
    & 3 & & 0.80  & & $N(0,1)*N(0,10^{-6})$ & & $N(0,1)*N(3.5, 0.5)$ \\
    & 4 & & 0.95  & & $N(0,1)*N(0,10^{-6})$ & & $N(0,1)*{\rm Gamma}(12, 0.25)$ \\
    & 5 & & 0.90  & & $N(0,1)*N(0,10^{-6})$ & & $N(0,1)*{\rm Gamma}(12, 0.25)$ \\
    & 6 & & 0.80  & & $N(0,1)*N(0,10^{-6})$ & & $N(0,1)*{\rm Gamma}(12, 0.25)$ \\
  \hline
    Bivariate & 1 & & 0.95  & & $N(\mathbf{0}_2,\Sigma)*N(\mathbf{0}_2, 10^{-6}I_2)$ & & $N(\mathbf{0}_2,\Sigma)*N((3.5,3.5)', 0.5I_2)$ \\
    & 2 & & 0.90  & & $N(\mathbf{0}_2,\Sigma)*N(\mathbf{0}_2, 10^{-6}I_2)$ & & $N(\mathbf{0}_2,\Sigma)*N((3.5,3.5)', 0.5I_2)$ \\
    & 3 & & 0.80  & & $N(\mathbf{0}_2,\Sigma)*N(\mathbf{0}_2, 10^{-6}I_2)$ & & $N(\mathbf{0}_2,\Sigma)*N((3.5,3.5)', 0.5I_2)$ \\
    & 4 & & 0.95  & & $N(\mathbf{0}_2,\Sigma)*N(\mathbf{0}_2, 10^{-6}I_2)$ & & $N(\mathbf{0}_2,\Sigma)*{\rm Gamma}_2(12, 0.25)$ \\
    & 5 & & 0.90  & & $N(\mathbf{0}_2,\Sigma)*N(\mathbf{0}_2, 10^{-6}I_2)$ & & $N(\mathbf{0}_2,\Sigma)*{\rm Gamma}_2(12, 0.25)$ \\
    & 6 & & 0.80  & & $N(\mathbf{0}_2,\Sigma)*N(\mathbf{0}_2, 10^{-6}I_2)$ & & $N(\mathbf{0}_2,\Sigma)*{\rm Gamma}_2(12, 0.25)$ \\
  \hline
  \end{tabular}

   \end{center}
\end{table}

\subsection{Simulations for univariate models}

We consider 6 scenarios for simulations. For the first 3 scenarios, we assume normal mixtures for null and alternative distributions with different $p_0$. The next 3 scenarios are the same as the first 3 except using a log-concave distribution as an alternative. 

Tables \ref{uni_p0est} and  \ref{uni_rmse_locfdr} summarize the simulation results. Table \ref{uni_p0est} shows that the proposed approach stably yields good $p_0$ estimates for a wide range of target values while Efron's method tends to give  overestimates $p_0$ especially when $p_0$ is relatively small. When the alternative follows a log-concave, the performance of our method are similar to that with a normal alternative. However, Efron method is slightly worse with the log-concave alternative. As a result, multiple testing procedure with Efron's method becomes too conservative when $p_0$ is small. Table \ref{uni_rmse_locfdr} reports similar results. The RMSEs of Efron's method  increase as $p_0$ decreases and are worse with the log-concave alternative.

Figure \ref{fig:uni1} illustrates the comparisons of the three performance measures in Scenario 1 with  various $fdr$ levels (from 0.05 to 0.25). Efron's method reports low FDR values regardless of $fdr$ values at the cost of higher FNR values. While there is no clear relationship between $fdr$ and FDR, \citet{Efron07} indicated $fdr = 0.20$ corresponds to FDR values between 0.05 and 0.15. Hence our method seems to control the FDR reasonably well given $fdr$ levels with lower FNR values. Similar patterns are founds in Figures \ref{fig:uni2} - \ref{fig:uni6} for other scenarios. 

 \begin{table}[t]
  \begin{center}
    \caption{\label{uni_p0est} Summary of $p_0$ estimation results from univariate models}

  \begin{tabular}{cccccccc}
  \hline
          &       & & \multicolumn{2}{c}{Proposed} & & \multicolumn{2}{c}{Efron} \\
          \cline{4-5} \cline{7-8}
    Scenario & $p_0$ & & Mean & standard error & & Mean & standard error \\
  \hline
    1 & 0.95  & & 0.9293 & 0.0006 & & 0.9490 & 0.0021 \\
    2 & 0.90  & & 0.8840 & 0.0006 & & 0.8971 & 0.0017 \\
    3 & 0.80  & & 0.7910 & 0.0007 & & 0.8878 & 0.0027 \\
    4 & 0.95  & & 0.9316 & 0.0006 & & 0.9577 & 0.0022 \\
    5 & 0.90  & & 0.8915 & 0.0006 & & 0.9169 & 0.0017 \\
    6 & 0.80  & & 0.8057 & 0.0008 & & 0.9215 & 0.0024 \\
  \hline
  \end{tabular}
  \end{center}
\end{table}

\begin{table}[b]
  \begin{center}
    \caption{\label{uni_rmse_locfdr} RMSE comparison of $fdr$ estimates for univariate models}

  \begin{tabular}{ccccccc}
  \hline
      & & \multicolumn{2}{c}{Proposed} & & \multicolumn{2}{c}{Efron} \\
                   \cline{3-4} \cline{6-7}
    Scenario & & Mean & standard error & & Mean & standard error \\
  \hline
    1 & & 0.1059 & 0.0017 & & 0.0728 & 0.0026 \\
    2 & & 0.0807 & 0.0015 & & 0.0712 & 0.0025 \\
    3 & & 0.0581 & 0.0013 & & 0.2520 & 0.0044 \\
    4 & & 0.1074 & 0.0018 & & 0.0891 & 0.0033 \\
    5 & & 0.0754 & 0.0016 & & 0.1185 & 0.0033 \\
    6 & & 0.0456 & 0.0013 & & 0.3730 & 0.0039 \\
  \hline
  \end{tabular}
  \end{center}
\end{table}

\begin{figure}[b]
   \centering
   \includegraphics[width=6in,keepaspectratio=TRUE]{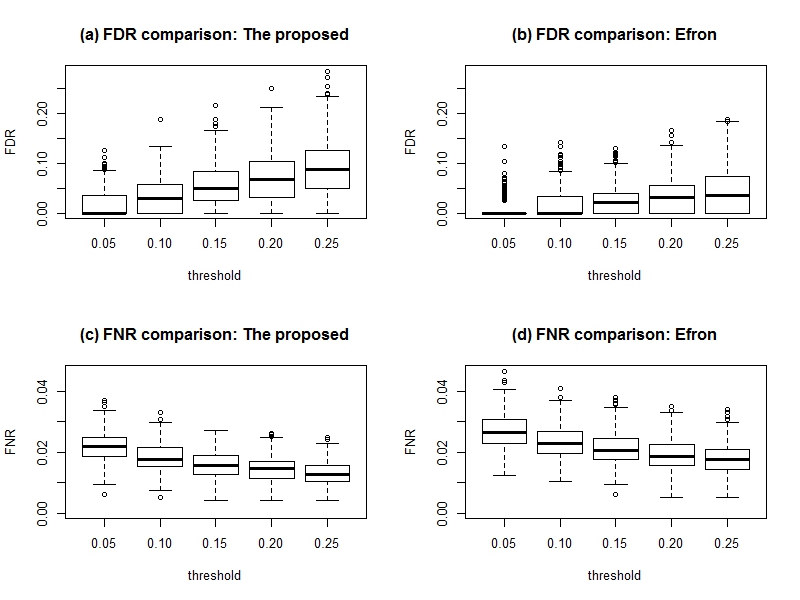} 
   \caption{Empirical False Discovery Rate and  empirical False Nondiscovery Rate from  the proposed method and Efron's method for Scenario 1 in the univariate simulations. The threshold are $fdr$ values. The boxplots are drawn with the values resulted from $M=500$ Monte Carlo simulations.}
   \label{fig:uni1}
\end{figure}

\begin{figure}[htbp]
   \centering
   \includegraphics[width=6in,keepaspectratio=TRUE]{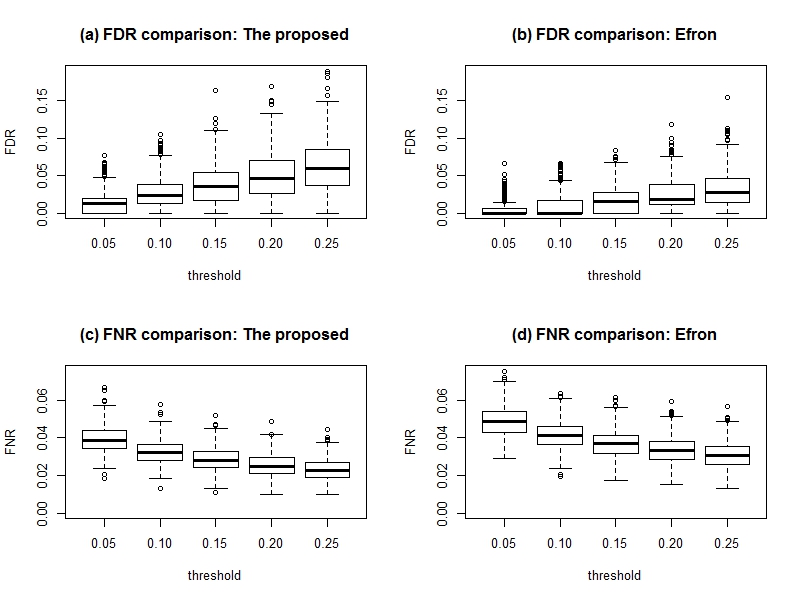} 
   \caption{Empirical False Discovery Rate and  empirical False Nondiscovery Rate from  the proposed method and Efron's method for Scenario 2 in the univariate simulations. The threshold are $fdr$ values. The boxplots are drawn with the values resulted from $M=500$ Monte Carlo simulations.}  
    \label{fig:uni2}
\end{figure}

\begin{figure}[htbp]
   \centering
   \includegraphics[width=6in,keepaspectratio=TRUE]{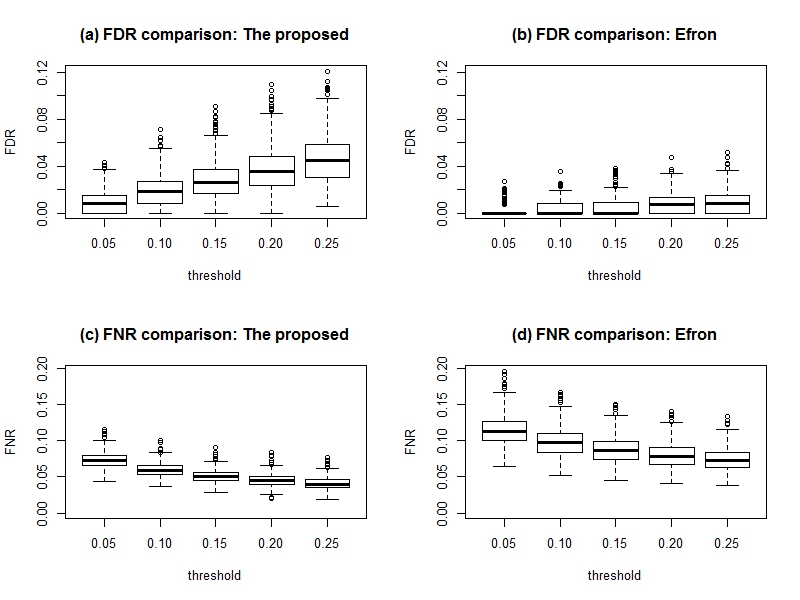} 
   \caption{Empirical False Discovery Rate and  empirical False Nondiscovery Rate from  the proposed method and Efron's method for Scenario 3 in the univariate simulations. The threshold are $fdr$ values. The boxplots are drawn with the values resulted from $M=500$ Monte Carlo simulations.}
   \label{fig:uni3}
\end{figure}

\begin{figure}[htbp]
   \centering
   \includegraphics[width=6in,keepaspectratio=TRUE]{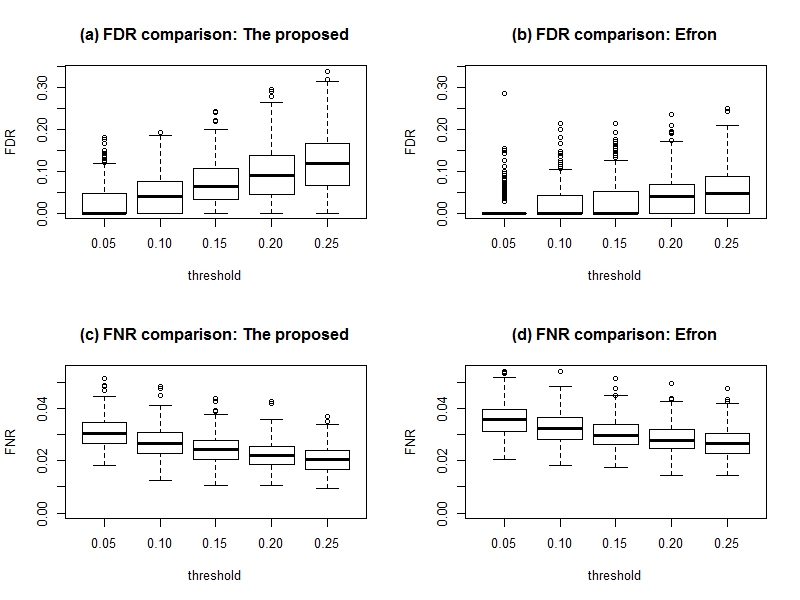} 
   \caption{Empirical False Discovery Rate and  empirical False Nondiscovery Rate from  the proposed method and Efron's method for Scenario 4 in the univariate simulations. The threshold are $fdr$ values. The boxplots are drawn with the values resulted from $M=500$ Monte Carlo simulations.}
   \label{fig:uni4}
\end{figure}

\begin{figure}[htbp]
   \centering
   \includegraphics[width=6in,keepaspectratio=TRUE]{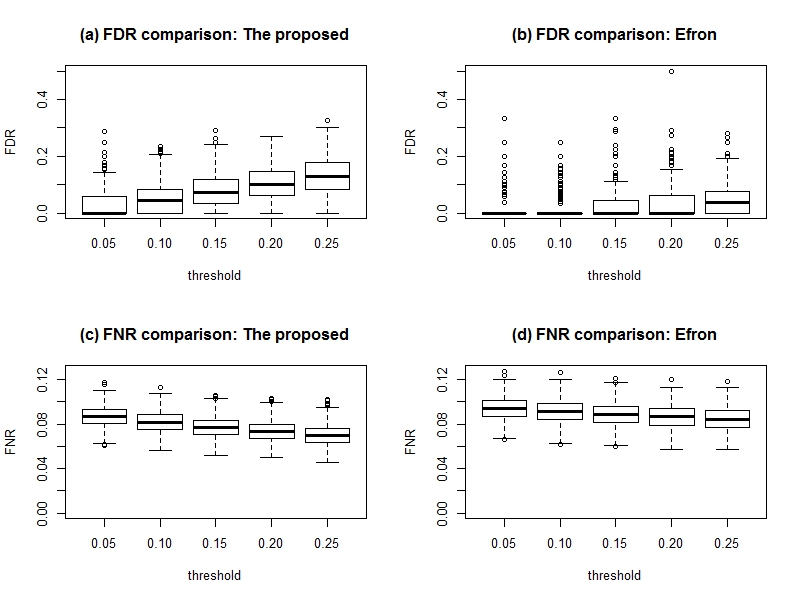} 
   \caption{Empirical False Discovery Rate and  empirical False Nondiscovery Rate from  the proposed method and Efron's method for Scenario 5 in the univariate simulations. The threshold are $fdr$ values. The boxplots are drawn with the values resulted from $M=500$ Monte Carlo simulations.}
   \label{fig:uni5}
\end{figure}

\begin{figure}[htbp]
   \centering
   \includegraphics[width=6in,keepaspectratio=TRUE]{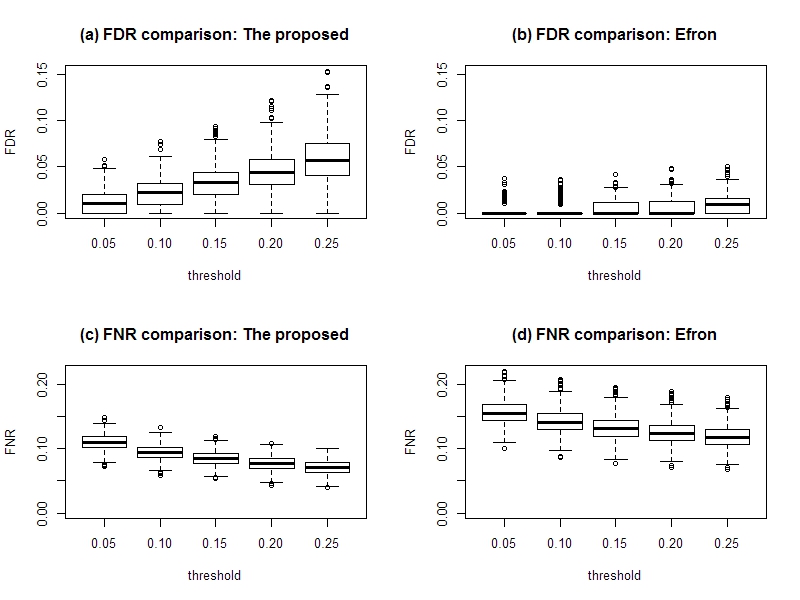} 
   \caption{Empirical False Discovery Rate and  empirical False Nondiscovery Rate from  the proposed method and Efron's method for Scenario 1 in the univariate simulations. The threshold are $fdr$ values. The boxplots are drawn with the values resulted from $M=500$ Monte Carlo simulations.}
   \label{fig:uni6}
\end{figure}

\subsection{Simulations for bivariate models}

We consider similar scenarios for bivariate models. Define  
$$
\mathbf{0}_2=\begin{pmatrix} 0 \\  0   \end{pmatrix}, \quad 
\mathbf{1}_2=\begin{pmatrix} 1 \\  1   \end{pmatrix}, \quad I_2 = \begin{pmatrix} 1 & 0 \\  0 &  1  \end{pmatrix}, \quad \Sigma = \begin{pmatrix} 1 & 0.3 \\ 0.3 &  1 \end{pmatrix}.
$$
The simulation scenarios are similar to those of univariate cases. In each scenario, the marginal distributions of null and alternative are the same as the null and alternative distributions of the corresponding  scenario in univariate models. By doing so, we want to find whether there is any advantage of 2-dimensional $fdr$ compared to 1-dimensional $fdr$. Like 

Figure \ref{fig:multi1} -- Figure \ref{fig:multi6} present the distributions of the $p_0$ estimates, empirical FDR and empirical FNR with various $fdr$  levels for each scenario. All the figures show that the proposed method works reasonably well. It is interesting to observe that the performance of the proposed method in bivariate models looks better than in univariate setting. Comparing the boxplots of FDR and FNR in panel (b) and (c) of the figures with those in the panels (b) and (c) of Figure \ref{fig:uni1} -- Figure \ref{fig:uni6}, we confirm that empirical FDR and FNR in the bivariate models are lower than those from the corresponding univariate models. Bivariate method could do better by taking account of extra information since it allows the test procedure an additional flexibility in deciding the boundary of the rejection region.

\begin{figure}[htbp]
   \centering
   \includegraphics[width=6in,keepaspectratio=TRUE]{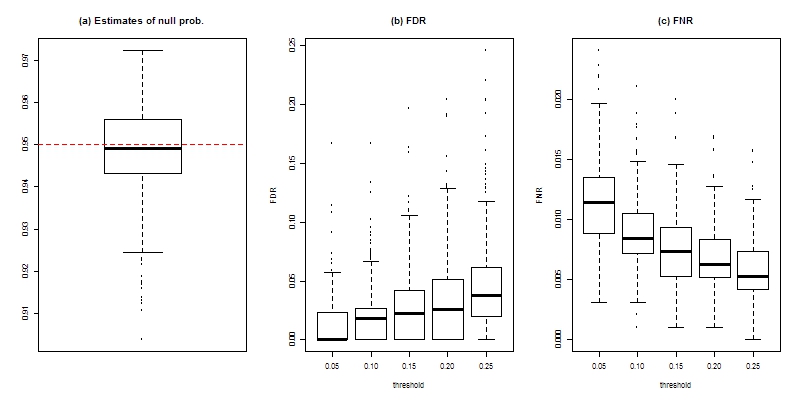} 
   \caption{(a) $p_0$ estimates, (b) empirical FDR,  (c) empirical FNR for Scenario 1 in the bivariate models with  $p_0=0.95$ and $f_1$ = Normal*Normal.}
   \label{fig:multi1}
\end{figure}

\begin{figure}[htbp]
   \centering
   \includegraphics[width=6in,keepaspectratio=TRUE]{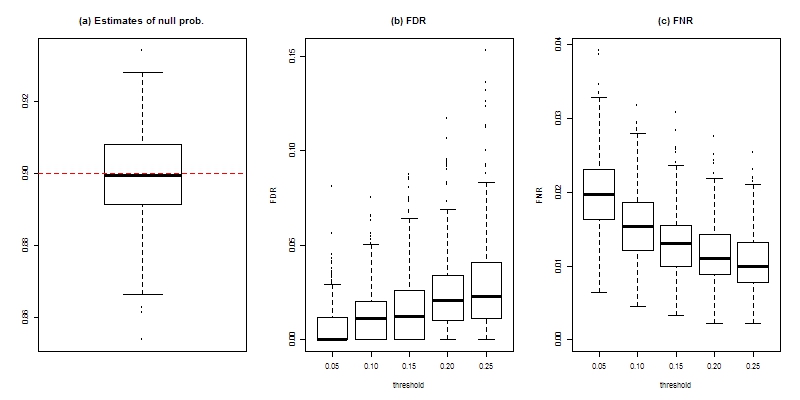} 
   \caption{(a) $p_0$ estimates, (b) empirical FDR,  (c) empirical FNR  for Scenario 2 in the bivariate models with  $p_0=0.90$ and $f_1$ = Normal*Normal.}
   \label{fig:multi2}
\end{figure}

\begin{figure}[htbp]
   \centering
   \includegraphics[width=6in,keepaspectratio=TRUE]{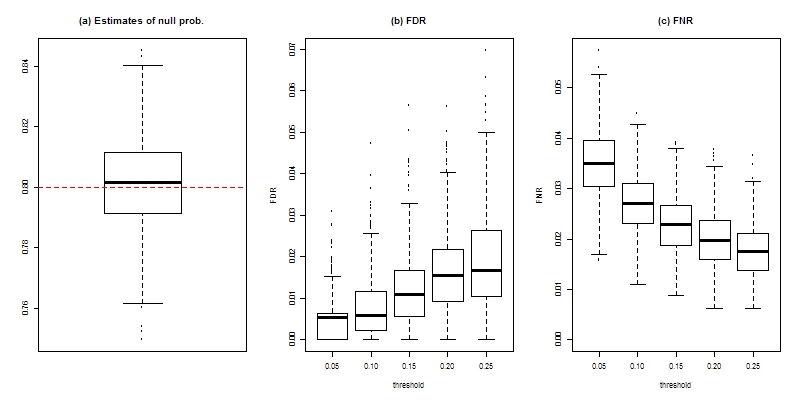} 
   \caption{(a) $p_0$ estimates, (b) empirical FDR,  (c) empirical FNR   for Scenario 2 in the bivariate models with  $p_0=0.80$ and $f_1$= Normal*Normal.}
   \label{fig:multi3}
\end{figure}

\begin{figure}[htbp]
   \centering
   \includegraphics[width=6in,keepaspectratio=TRUE]{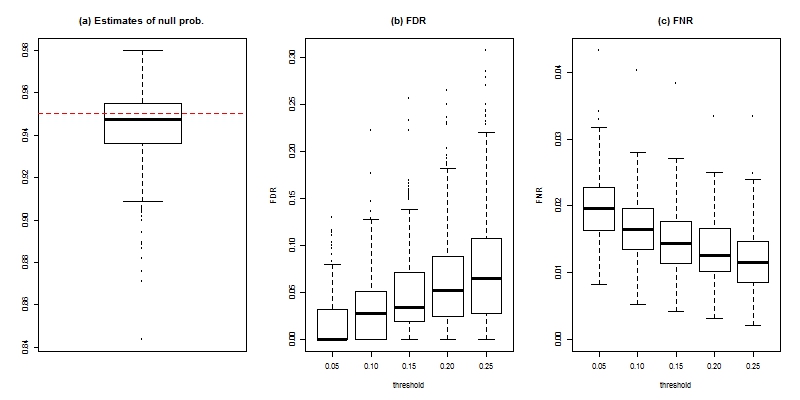} 
   \caption{(a) $p_0$ estimates, (b) empirical FDR,  (c) empirical FNR   for Scenario 2 in the bivariate models with 
   T  $p_0=0.95$ and $f_1$= Normal*Gamma.}
   \label{fig:multi4}
\end{figure}

\begin{figure}[htbp]
   \centering
   \includegraphics[width=6in,keepaspectratio=TRUE]{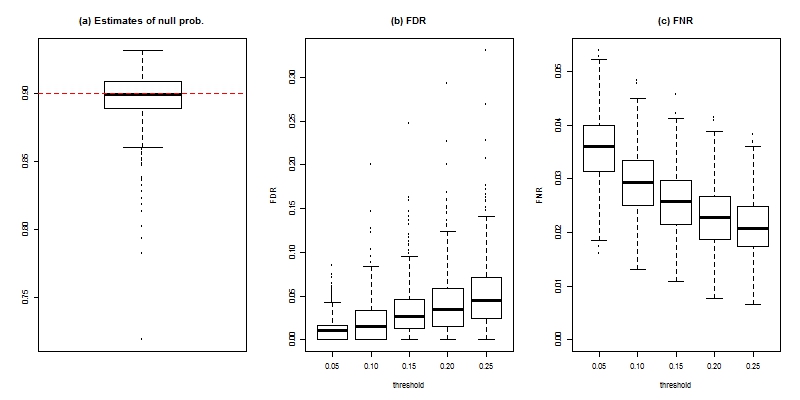} 
   \caption{(a) $p_0$ estimates, (b) empirical FDR,  (c) empirical FNR  for Scenario 2 in the bivariate models with  $p_0=0.90$ and $f_1$= Normal*Gamma.}
   \label{fig:multi5}
\end{figure}

\begin{figure}[htbp]
   \centering
   \includegraphics[width=6in,keepaspectratio=TRUE]{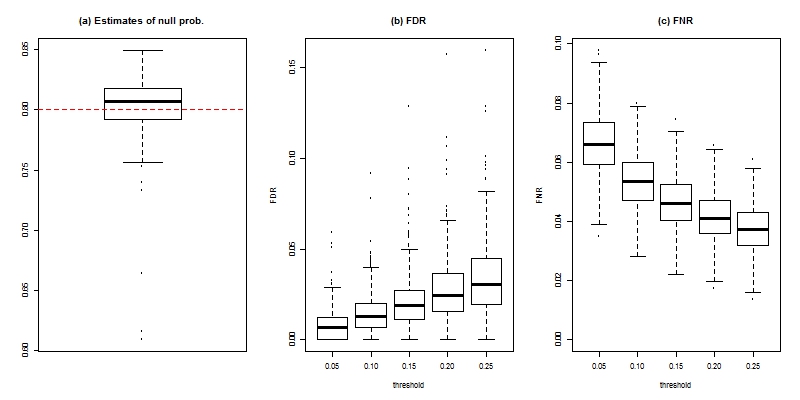} 
   \caption{(a) $p_0$ estimates, (b) empirical FDR,  (c) empirical FNR   for Scenario 2 in the bivariate models with  $p_0=0.80$ and $f_1$= Normal*Gamma.}
   \label{fig:multi6}
\end{figure}

 \section{Case Studies}

 \subsection{Ophthalmologic gene expression data analysis}
 
 In this subsection, we  consider  an Affymetrix data set with  a paired design \citep{Smith07, Choi08}.  Donor-matched retina and choroidal endothelial cells from three human cadavers were treated with Toxoplasma gondii tachyzoite or medium alone. We are interested in treatment effects. It measured gene expression levels in retina and choroid. Two tissues are next to each other, yet quite unique.  Gene expression profiling experiment were performed by Affymetrix Human Genome Focus array that contains 8746 human gene transcripts. More details on experiment and normalization can be found in \citet{Smith07}.  
 
 Let  $y_{ij}$ be $j$th gene expression value of $i$th subject. A simple linear model is fitted to the normalized data to find out significantly changed genes in each cell type.

\begin{eqnarray*}
&\mbox{Model 1 (retina):}&      y_{ij} = \beta_0 +  \beta_t \cdot \mbox{treatment} + \beta_{i}\cdot\mbox{subject}_i + \epsilon_{ij}, \\
&\mbox{Model 2 (choroid):}&      y_{ij} = \beta_0 +  \beta_t \cdot \mbox{treatment} + \beta_{i}\cdot\mbox{subject}_i + \epsilon_{ij}. 
\end{eqnarray*}

An empirical Bayes implemented in {\sf R} package limma \citep{Ritchie15} is used to fit the models to the normalized data.  With 1-D $fdr <0.05$, there are 248 and 173 significantly changed probe sets for the treatment group in Model 1 and 2, respectively. There are 68 common probe sets detected by both models and in total 353 probe sets are detected by at least one of them. In comparison, when applying the proposed 2-D $fdr$, there are 713 significant probe sets with 2-D $fdr <0.05$. The proposed 2-D  $fdr$ could detect more than twice of probe sets that are detected by 1-D  $fdr$ applied separately.

For comparison, another linear model that is more likely used in practice for the study is fitted with 1-D $fdr$. Now let $y_{ijk}$ be the $j$th gene expression level of $i$th subject in $k$th cell type.  
$$
\mbox{Model 3 :}  y_{ijk} = \beta_0 + \beta_t \cdot \mbox{treatment} + \beta_i \cdot \mbox{subject}_i + \beta_c \cdot \mbox{cell type} + \epsilon_{ijk} 
$$
Four hundred one probe sets had  $fdr <0.05$ less than 0.05, which is somewhat closer to 353 significant probe sets detected by 1-D $fdr$ apply to each cell type separately. When crosstabulated with 713 significant probe sets based 2-D $fdr$, there are 164 common probe sets between the two approaches.

A gene ontology analyses are performed by NIH DAVID (Database for Annotation, Visualization and Integrated Discovrey) to compare 713 significant probe sets based on 2-D $fdr$ to 401 based on 1-D $fdr$  of model 3. The former list returned 108 chart records of functional annotation with  $fdr< 0.05$. Interestingly, the analysis also detects that the list is abundant with genes from eye ($p=0.020$) and t-cell ($p=0.047$).  In comparison, similar analyses of the latter list of 401 probe sets returns only 85 significant chart records and $p$-values for abundance of genes in eye or t-cell are not significant (0.113 or 0.075, respectively).

 \subsection{Radial velocity distribution in astronomy}
 
 The radial velocity of a star play a key role in detecting exoplanet and is of great interest in astronomy \citep{Fabrizio11}.  The data are radial velocity measurements of 1266 luminous stars mainly from the Carina, a dwarf spheroidal galaxy and a companion of the Milky Way \citep{Patra16}. Some of data are  also from the Milky Way. Astronomers are interested in estimating the distribution of radial velocities that is believed to be a two-component mixture where one is a Gaussian and the other is a non-Gaussian. We are particularly interested in understanding the Gaussian component of the Radial Velocity distribution.  
 
 This data set was analyzed by \citet{Patra16}. They reported that the estimate of $p_0$=0.36 and the estimated location of the peak (=$\mu_0$) is 229.1 with the estimated stand deviation is 7.51.  Figure \ref{fig:rv} shows the mixture of two components estimated by the proposed method. The right one is considered as the Gaussian part \citep{Fabrizio11}. Our method presents $\hat p_0 = 0.37$, $\hat \mu_0 = 223.5$ and $\hat \sigma_0 = 10.7$.  With slightly different  data sets, \citet{Fabrizio11} reported the peaks around 220.4-221.4 with the standard deviation 8.0-11.7.

 \begin{figure}[ht]
   \centering
   \includegraphics[width=6in,keepaspectratio=TRUE]{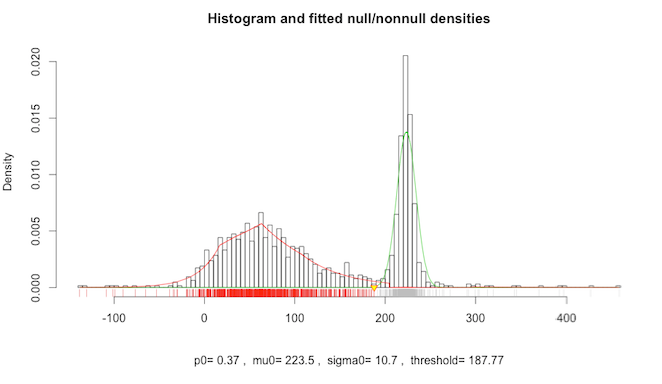} 
   \caption{Density of the Redial Velocity Distribution}
   \label{fig:rv}
\end{figure}
 
\section{Conclusion}

In this paper we proposed a semiparametric two component mixture model based on  normal and log concave densities. Our method is flexible  because it does not require smoothing parameter selection and  the log-concavity assumption of the alternative is reasonable. We discuss the connection between the proposed model and multiple testing. Our method estimate $fdr$ and the proportion of the null simultaneously and fit the alternative when necessary. 

We presented an EM-type algorithm to implement the proposed estimators and can be easily extended to multivariate settings. Our simulation studies suggested that the proposed method outperforms other existing method in both univariate and bivariate settings.  We present two  case studies which shows  the flexibility of the semiparametric mixture model in estimating Radial Velocity distributions in astronomy and  advantages of using 2-d $fdr$ over 1-d $fdr$ in microarray.
 
  In general, semiparametric mixture models are nonidentifiable without additional assumptions on $f_1$. Assuming that $f_1$ belongs to the location-shift family of distributions, \citet{Bordes06} showed the identifiability of the semiparametric mixture model under mild regularity conditions. \citet{Genovese04} also addressed identifiability for mixture models for $p$-value distribution under the assumption that the null is uniform. We left the identifiability of our mixture model as a future work.

\section*{Acknowledgement} 
The authors thanks Professor Bodhisattva Sen and Dr. Rohit Kumar Patra for sharing the Carina data and their codes.

\end{document}